\documentclass{article}

\usepackage{PRIMEarxiv}
\usepackage[utf8]{inputenc}
\usepackage[T1]{fontenc}
\usepackage{url}
\usepackage{booktabs}
\usepackage{amsfonts}
\usepackage{nicefrac}
\usepackage{microtype}
\usepackage{amsmath}
\usepackage{fancyhdr}
\usepackage{graphicx}
\usepackage{hyperref}
\usepackage{amsthm}

\theoremstyle{plain}
\newtheorem{theorem}{Theorem}
\newtheorem{proposition}[theorem]{Proposition}

\theoremstyle{definition}

\theoremstyle{remark}
\graphicspath{{media/}}

\title{Wave-Based Semantic Memory with Resonance-Based Retrieval: A Phase-Aware Alternative to Vector Embedding Stores}

\author{
  Aleksandr Listopad \\
  Independent Researcher \\
  \texttt{al@evacortex.ai}
}

\begin{document}
\maketitle

\begin{abstract}
Conventional vector-based memory systems rely on cosine or inner product similarity within real-valued embedding spaces. While computationally efficient, such approaches are inherently phase-insensitive and limited in their ability to represent contextual modulation, polarity, or structured semantic transformations.

We introduce a wave-based memory representation, in which embedding vectors are transformed into fixed-length complex-valued waveforms of the form $\psi(x) = A(x)\,e^{i\phi(x)}$. Here, $x$ indexes vector dimensions, $A(x)$ encodes semantic amplitude, and $\phi(x)$ encodes contextual phase. This formulation supports a phase-aware similarity function---the \emph{resonance score}---which reflects alignment in both amplitude and phase and intuitively quantifies constructive interference between semantic patterns.

Formally, the resonance score is computed as the squared magnitude of the elementwise complex sum of two patterns, aggregated across dimensions and normalized by their combined energy. The scale-alignment factor $R$ penalizes energy imbalance and prevents dominance by high-intensity inputs; the full expression keeps $S$ within $[0,1]$. This metric generalizes cosine similarity by incorporating phase coherence, enabling comparison between meaning-modulated representations that would otherwise appear similar in traditional spaces.

We implement this model in \emph{ResonanceDB}, a source-available system that stores amplitude--phase patterns in memory-mapped binary segments and evaluates similarity using a deterministic comparison kernel. Compatibility with standard vector embeddings is preserved via a \emph{sign--phase mapping} ($A(x)=|v(x)|$, $\phi(x)=0$ if $v(x)\ge 0$, $\phi(x)=\pi$ otherwise), which maintains a valid polar form while retaining sign information. For vectors with non-negative entries, a simple zero-phase initialization ($\phi(x)=0$) is also valid.

Empirical evaluation on synthetic and embedding-derived datasets shows that phase-enriched queries improve top-$k$ retrieval in tasks involving negation, inversion, and contextual shift---distinctions often blurred under cosine-based retrieval. These results suggest that wave-based memory provides a cognitively inspired, phase-sensitive alternative to conventional vector stores, expanding the expressive capacity of semantic retrieval in reasoning-oriented applications.
\end{abstract}


\section{Introduction}
Vector embeddings are the dominant paradigm for encoding semantic information in modern AI systems. From retrieval-augmented generation (RAG) pipelines to knowledge graph embedding and neural search engines, high-dimensional vectors derived from pre-trained language models are widely used to approximate meaning through geometric proximity, typically relying on cosine distance or inner product~\cite{johnson2019faiss, jegou2011product}. These representations are computationally efficient, differentiable, and compatible with approximate nearest neighbor (ANN) methods.

Despite their widespread adoption, vector-based embeddings exhibit fundamental limitations in tasks involving semantic polarity, contextual shifts, or compositional reasoning. Standard embedding spaces represent concepts as static points, but lack mechanisms to encode transformations such as negation, modality, or epistemic stance. For example, embeddings for ``happy'' and ``not happy'' are often closely aligned, despite having opposite meaning~\cite{mikolov2013linguistic}. The geometry captures surface similarity, but fails to reflect operator-level modulation.

To address this, we introduce a wave-based memory representation in which semantic patterns are modeled as complex-valued waveforms of the form $\psi(x) = A(x)\,e^{i\phi(x)}$, where $A(x)$ encodes semantic amplitude and $\phi(x)$ encodes contextual phase. Here, $x$ denotes the index of a vector dimension. This formulation draws on constructive interference, treating meaning not as a static vector but as a modulated wave pattern whose alignment with other patterns depends on both magnitude and phase. Operations such as negation, modality, or discourse shift correspond to phase transformations, while similarity is preserved through phase-sensitive interference energy.

We implement this model in a practical system, ResonanceDB, which stores and retrieves wave patterns as amplitude--phase pairs using memory-mapped binary segments. Similarity in ResonanceDB is evaluated via a deterministic comparison kernel that computes constructive interference energy between two complex patterns. This energy is maximal when patterns are in-phase and proportionally aligned, and minimal when anti-phased.

Importantly, ResonanceDB is compatible with standard vector embeddings: any real-valued vector can be mapped to a wave pattern via amplitude assignment and sign--phase initialization, enabling seamless integration into existing pipelines without retraining or re-indexing.

The goal of this work is not to replace vector embeddings, but to extend their expressiveness by introducing a complementary, phase-aware semantic substrate. Our contribution lies in formalizing a resonance-based similarity metric, implementing an efficient storage and retrieval system, and demonstrating empirically that this framework captures distinctions inaccessible to conventional vector methods.

\section{WavePattern Architecture}

\begin{figure}[htbp]
\centering
\includegraphics[width=0.80\linewidth]{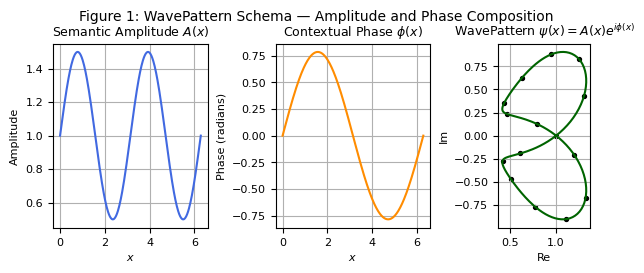}
\caption{WavePattern schema. Top: amplitude $A(x)$. Middle: contextual phase $\phi(x)$. Bottom: complex trajectory $\psi(x) = A(x)\,e^{i\phi(x)}$ in the complex plane.}
\label{fig:wavepattern}
\end{figure}

Each semantic pattern is modeled as a discrete complex-valued waveform
\begin{equation}
\psi(x) = A(x)\,e^{i\phi(x)},
\end{equation}
where $x$ indexes vector dimensions, $A(x)$ is the real-valued amplitude encoding semantic intensity ($A(x)\ge 0$), and $\phi(x)$ is the phase component representing contextual modulation ($\phi(x)\in[-\pi,\pi)$). The resulting $\psi(x)$ is a pointwise complex representation of meaning, where amplitude governs salience and phase expresses structural or operator-level semantics.

Similarity between two waveforms $\psi_1(x)$ and $\psi_2(x)$ is computed using a resonance-based metric derived from constructive interference:
\begin{equation}
\label{eq:resonance}
S(\psi_1,\psi_2) \;=\; \tfrac{1}{2}\cdot
\frac{\sum_x \left|\psi_1(x)+\psi_2(x)\right|^2}{\sum_x \left(|\psi_1(x)|^2+|\psi_2(x)|^2\right)} \cdot R,
\end{equation}
where the scale-alignment factor $R$ is
\begin{equation}
\label{eq:R}
R \;=\; \frac{2\sqrt{E_1E_2}}{E_1+E_2}, \qquad
E_1=\sum_x|\psi_1(x)|^2,\;\; E_2=\sum_x|\psi_2(x)|^2.
\end{equation}
If $E_1+E_2=0$, we define $S=0$.

\begin{figure}[htbp]
\centering
\includegraphics[width=0.5\linewidth]{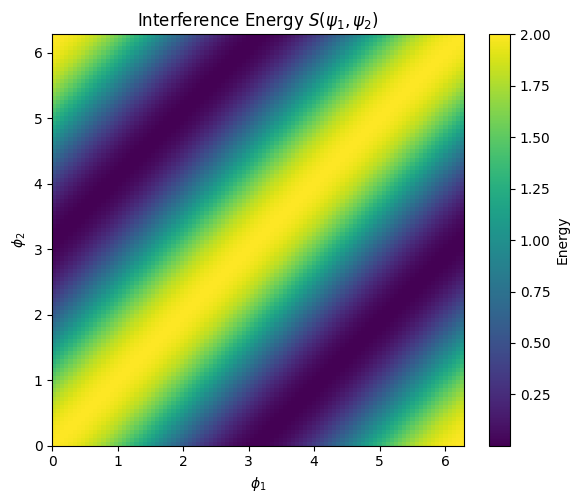}
\caption{Resonance score $S(\psi_1, \psi_2)$ as a function of relative phase $\delta=\phi_2-\phi_1$. Maximal resonance occurs at in-phase alignment ($\delta=0$), minimal at anti-phase ($\delta=\pi$).}
\label{fig:interference}
\end{figure}

This scaling emphasizes interference between patterns of comparable energy and prevents dominance by high-intensity inputs. With purely real patterns ($\Im\psi=0$) and equal norms ($E_1=E_2$), Equation~(\ref{eq:resonance}) reduces to $\tfrac{1+\cos\theta}{2}$, where $\theta$ is the angle between real vectors. For general vectors with arbitrary phases, $S$ remains bounded in $[0,1]$ and symmetric.

Patterns are stored as real-valued pairs $(A,\phi)$ and converted to complex form on demand. The system provides two backends: a scalar Java backend (\texttt{JavaKernel}, CPU baseline) and an \emph{experimental} SIMD backend (\texttt{SimdKernel}) using the Panama Vector API (Java~22).

\section{Comparison with Vector Stores}
Conventional vector similarity uses cosine similarity:
\begin{equation}
\text{cosine}(u,v)=\frac{u\cdot v}{\|u\|\,\|v\|}.
\end{equation}

\begin{figure}[htbp]
\centering
\includegraphics[width=0.8\linewidth]{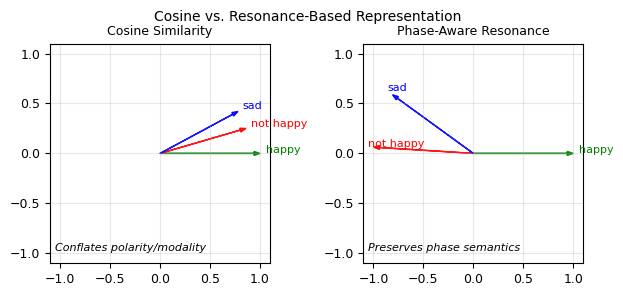}
\caption{Illustration: Cosine similarity (left) can conflate polarity (e.g., ``happy'', ``not happy'', ``sad'' may appear close). Resonance-based space (right) distinguishes them via phase.}
\label{fig:cosine_vs_resonance}
\end{figure}

While efficient, cosine is phase-insensitive and may conflate contrasts (e.g., negation or epistemic shifts) that are structurally distinct.

Resonance operates on complex-valued waveforms and reflects both amplitude alignment and phase coherence. For $\psi_1,\psi_2$ derived from $(A_1,\phi_1)$ and $(A_2,\phi_2)$, $S(\psi_1,\psi_2)$ satisfies:
\begin{align*}
& S \in [0,1] && \text{(bounded)} \\
& S(\psi_1,\psi_2)=S(\psi_2,\psi_1) && \text{(symmetric)} \\
& S(\psi,\psi)=1 && \text{(self-match)} \\
& S(\psi_1,\psi_2)=1 \ \text{iff}\ \psi_1=\psi_2\neq 0 && \text{(self-match only)} \\
& S(\psi_1,\psi_2)=0 \ \text{when}\ \psi_2(x)=-\,\psi_1(x) && \text{(anti-phase)} \\
& S(e^{i\delta}\psi_1,e^{i\delta}\psi_2)=S(\psi_1,\psi_2) && \text{(global phase invariance)}.
\end{align*}
Mapping any real vector $v\!\in\!\mathbb{R}^L$ to $A(x)\!=\!|v(x)|,\;\phi(x)\!=\!0$ if $v(x)\ge 0$ else $\pi$ yields cosine-like behavior: for equal norms $S=\tfrac{1+\cos\theta}{2}$; in general, for fixed norms $S$ is monotonic in $\cos\theta$ only in the purely real case.

\section{Experimental Evaluation}
\label{sec:experiments}
We evaluate (i) latency and scalability, (ii) sensitivity to amplitude/phase variation, and (iii) retrieval precision under structured semantic perturbations.

\subsection{Setup}
The experiments ran on a single machine with \textbf{Windows 11 (10.0)}, \textbf{Java HotSpot 22.0.1} (64-bit), a \textbf{12-core Intel CPU} (3.60\,GHz), and \textbf{32\,GB RAM}, no GPU. The patterns had fixed length $L\in\{512,1024\}$; dataset sizes ranged from $10^4$ to $5\times 10^5$. Embedding-derived inputs were mapped to amplitudes with sign--phase initialization; synthetic patterns varied phase to isolate phase effects. Top-$k$ queries ($k\in\{1,10,100\}$) were executed under \emph{cold-cache} conditions using a fixed \textbf{12-thread} Java executor with concurrent scans over memory-mapped segments.

\textbf{All latency results below use the scalar backend \texttt{JavaKernel} (no SIMD). These figures represent a conservative baseline; the SIMD backend (\texttt{SimdKernel}, Panama Vector API) is experimental and excluded.}

\subsection{Operator Retrieval Metrics}
On the compact operator set (bases \emph{happy} and \emph{good}), resonance achieves \textbf{P@1 = 1.0} for NEG, SHIFT+, INT\_UP, and INT\_DOWN, whereas cosine yields \textbf{P@1 $\approx$ 0.0} across these operators in this setup. These results are consistent with the distributional and heatmap analyses below.

\subsection{Qualitative Analysis of Semantic Operators}
\label{sec:qualitative}
To illustrate the difference between cosine and resonance similarity, we analyzed four representative operators applied to base terms such as \emph{happy} (and also \emph{good} in our compact set): logical negation (NEG), controlled phase shift (SHIFT+), and intensity modulation (INT\_UP, INT\_DOWN).

\subsubsection{Distance Distributions}
For visualization, we report distances on $[0,1]$ defined as $d_\text{cos}=\frac{1-\mathrm{cosine}}{2}$ and $d_\text{res}=1-S$.
Figure~\ref{fig:operator_dists} shows the distance histograms under both metrics. Cosine produces nearly identical distributions for all operators, concentrated around $\approx 0.25$ (mean distance with std $\approx 0.08$). This collapse reflects the inability of cosine to separate negation, shifts, or intensity changes from the base term. Resonance, in contrast, yields operator-specific bands with distinct means (NEG $\approx 0.52$, SHIFT+ $\approx 0.40$, INT\_UP $\approx 0.45$, INT\_DOWN $\approx 0.42$) and higher variance ($\sigma \approx 0.10$--$0.15$), reflecting structured semantic differences. P@1 corroborates this separation.

\begin{figure}[htbp]
  \centering
  \includegraphics[width=0.8\linewidth]{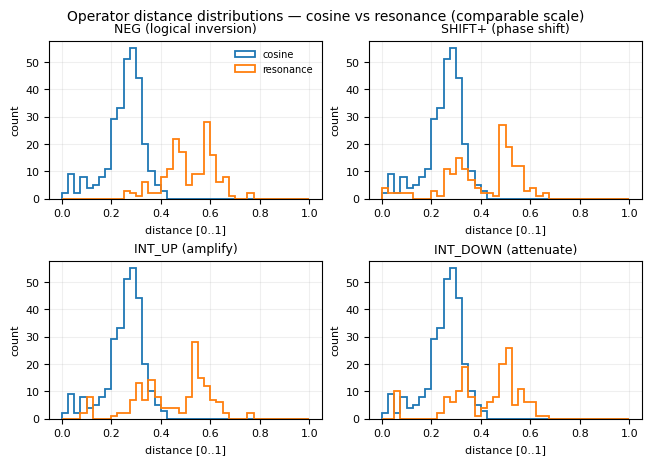}
  \caption{Operator histograms. Cosine distances collapse near $\approx 0.25$ for all operators, while resonance separates them into distinct, semantically consistent bands.}
  \label{fig:operator_dists}
\end{figure}

\subsubsection{Heatmap Structure}
Heatmaps (Figure~\ref{fig:heatmaps}) further highlight the contrast. Cosine distances form dense, unstructured clusters, where negated or intensified forms remain close to their bases. Resonance produces block-diagonal manifolds: NEG aligns with anti-phased opposites, SHIFT+ forms stable offset clusters, and INT\_UP/INT\_DOWN separate into amplitude-driven zones. This structured geometry demonstrates that resonance retrieval preserves operator semantics rather than collapsing them into undifferentiated proximity.

\begin{figure}[htbp]
  \centering
  \includegraphics[width=0.82\linewidth]{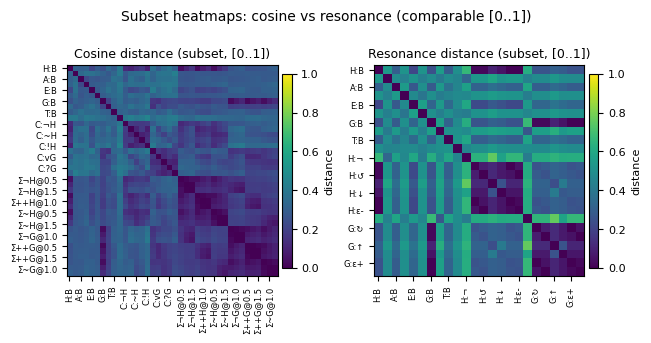}
  \caption{Side-by-side heatmaps. Cosine conflates operator variants near their bases, whereas resonance yields block structure aligned with semantic operators.}
  \label{fig:heatmaps}
\end{figure}

\subsection{Latency and Scalability}
On $L=1024$ and $500$K patterns (heap $\approx$ 512\,MB), we observe interactive CPU-only latencies (cold cache, 12 threads) with \texttt{JavaKernel}:
\begin{itemize}
  \item \textbf{top-1:} avg $\approx$ 135\,ms,\; p95 $\approx$ 152\,ms
  \item \textbf{top-10:} avg $\approx$ 138\,ms,\; p95 $\approx$ 156\,ms
  \item \textbf{top-100:} avg $\approx$ 145\,ms,\; p95 $\approx$ 163\,ms
\end{itemize}

\begin{figure}[htbp]
\centering
\includegraphics[width=0.7\linewidth]{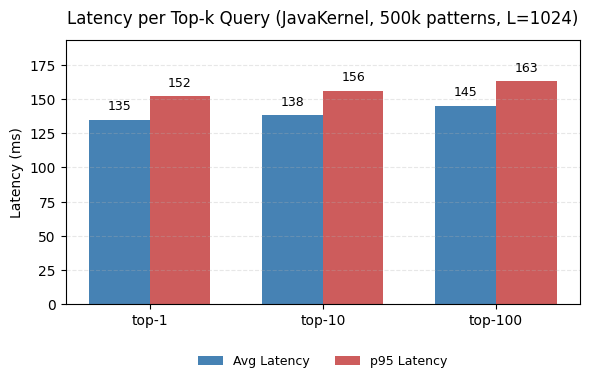}
\caption{Latency per top-$k$ query using JavaKernel ($L=1024$, 500K patterns, 12 threads, cold cache). Average and 95th percentile (p95) latencies scale linearly with dataset size and grow modestly with $k$.}
\label{fig:latency}
\end{figure}

These measurements demonstrate that full waveform comparison---including per-element phase conversion and interference accumulation---is compatible with \emph{interactive response times} on commodity CPUs, without indexing or approximation. Increased I/O throughput and SIMD acceleration are expected to reduce latency further.

\section{Related Work}
\subsection{Vector Similarity Search}
The prevailing approach relies on vector embeddings and geometric similarity metrics. Libraries such as FAISS~\cite{johnson2019faiss} implement efficient ANN search using quantization and indexing, typically grounded in cosine or inner product metrics~\cite{jegou2011product}. While scalable, they treat embeddings as static points, lacking explicit operators for polarity, modality, or compositional recontextualization.

\subsection{Holographic and Complex Embeddings}
Holographic models~\cite{nickel2016holographic} represent entities/relations via circular correlation inspired by associative memory. Complex-valued embeddings~\cite{trabelsi2018deep} introduce phase as a latent component; related phase-parameterized approaches include \emph{ComplEx}~\cite{trouillon2016complex} and \emph{RotatE}~\cite{sun2019rotate}, which model relations as rotations in the complex plane. Our approach elevates phase to a primary operator of semantic alignment, enabling interference-based retrieval sensitive to phase coherence.

\subsection{Memory-Augmented Neural Architectures}
Neural Turing Machines and Differentiable Neural Computers~\cite{graves2016hybrid} add external, differentiable memory accessed by soft attention. ResonanceDB is non-trainable: fixed patterns, deterministic comparison, no runtime parameter updates---closer to content-addressable memory than to a learnable controller.

\subsection{Cognitive and Neuromorphic Models of Meaning}
Interference-based encodings have roots in cognitive and neuromorphic models. Frameworks such as Plate’s HRR~\cite{plate2003holographic} and Kanerva’s Hyperdimensional Computing~\cite{kanerva2009hyperdimensional} embed symbolic structure in high-dimensional vectors via algebraic or convolutional operations. Many remain theoretical or tied to specialized hardware; our system delivers a pragmatic implementation on conventional infrastructure.

\section{Limitations}
ResonanceDB deliberately forgoes training: phases are supplied by the encoder or specified procedurally, and there is no learnable phase modulator. This design improves interpretability but may limit adaptivity. Because similarity depends on phase coherence, the method can be sensitive to phase noise or high-frequency perturbations; smoothing or regularization may be required. The present retrieval is exact and scan-based, with complexity $O(NL)$ over $N$ patterns of length $L$; at larger scales, memory-mapped I/O and per-query phase conversion can dominate latency. Although the scale-alignment factor $R$ mitigates imbalance, amplitude calibration across heterogeneous corpora remains an engineering consideration. The SIMD backend (\texttt{SimdKernel}, Panama Vector API) is experimental; reported latencies use the scalar Java backend.

\section{Discussion and Conclusion}
We proposed a wave-based memory architecture with a resonance-based similarity that is sensitive to both magnitude and phase. Phase shifts naturally encode contrastive phenomena (negation, modality, role inversion) that are difficult to express in traditional vector geometries. Empirically, phase-enriched queries improve top-$k$ retrieval on tasks where cosine fails to discriminate, while maintaining \emph{interactive response times} on CPUs.

From an engineering perspective, ResonanceDB attains practical latency through fixed-length patterns, memory-mapped segments, and efficient accumulation; existing embeddings are compatible via sign--phase mapping, enabling plug-in use within RAG and search pipelines. All reported measurements are CPU-only with the scalar backend; SIMD acceleration and improved I/O throughput are expected to further reduce latency. Thus, wave-based memory is promising not only conceptually but also as an engineering substrate for reasoning systems, where phase can encode roles, hypotheses, or epistemic status.

A source-available prototype of ResonanceDB is available for academic use under the \textbf{Prosperity Public License 3.0}.\footnote{\url{https://prosperitylicense.com/}} This license permits free use until revenue is generated, after which a commercial license is required.

\appendix
\section*{Appendix A: Formal Properties of the Resonance Score}

\paragraph{Notation.}
Let $\psi_k : \{1,\dots,L\}\to\mathbb{C}$ with $\psi_k(x)=A_k(x)e^{i\phi_k(x)}$, $A_k(x)\ge 0$.
We define the Hermitian inner product on $\mathbb{C}^L$ by
\[
\langle \psi_1,\psi_2\rangle \;=\; \sum_{x=1}^L \psi_1(x)\,\overline{\psi_2(x)},
\quad
E_k \;=\; \|\psi_k\|_2^2 \;=\; \sum_x |\psi_k(x)|^2 \;\ge 0.
\]
This inner product is linear in the first argument and conjugate-linear in the second.

Let
\[
S(\psi_1,\psi_2)
\;=\;
\frac{1}{2}\cdot
\frac{\sum_x|\psi_1(x)+\psi_2(x)|^2}{\sum_x\big(|\psi_1(x)|^2+|\psi_2(x)|^2\big)}
\cdot
\frac{2\sqrt{E_1E_2}}{E_1+E_2},
\]
with the convention $S=0$ if $E_1+E_2=0$.

\paragraph{Equivalent algebraic form.}
For complex numbers $z_1,z_2$ one has $|z_1+z_2|^2 = |z_1|^2 + |z_2|^2 + 2\,\Re\!\big(z_1\overline{z_2}\big)$.
Summing pointwise over $x$ yields
\[
\sum_x|\psi_1+\psi_2|^2 \;=\; E_1 + E_2 + 2\,\Re\langle\psi_1,\psi_2\rangle.
\]
Hence, whenever $E_1+E_2>0$,
\begin{equation}
\label{eq:S-compact}
S(\psi_1,\psi_2)
\;=\;
\frac{\big(E_1+E_2+2\,\Re\langle\psi_1,\psi_2\rangle\big)\,\sqrt{E_1E_2}}{(E_1+E_2)^2}.
\end{equation}

\begin{theorem}[Bounds]
For all $\psi_1,\psi_2$, $0\le S(\psi_1,\psi_2)\le 1$.
\end{theorem}

\begin{proof}
By Cauchy–Schwarz for Hermitian inner products,
$|\langle\psi_1,\psi_2\rangle|\le \sqrt{E_1E_2}$, hence
$-\sqrt{E_1E_2} \le \Re\langle\psi_1,\psi_2\rangle \le \sqrt{E_1E_2}$.
Using \eqref{eq:S-compact} we get
\[
S \;=\; \frac{\sqrt{E_1E_2}}{(E_1+E_2)^2}\,\Big(E_1+E_2+2\,\Re\langle\psi_1,\psi_2\rangle\Big).
\]
The minimum is at $\Re\langle\psi_1,\psi_2\rangle=-\sqrt{E_1E_2}$:
\[
S_{\min}
= \frac{\sqrt{E_1E_2}}{(E_1+E_2)^2}\,\big(E_1+E_2-2\sqrt{E_1E_2}\big)
= \frac{(\sqrt{E_1}-\sqrt{E_2})^2\,\sqrt{E_1E_2}}{(E_1+E_2)^2}\;\ge 0.
\]
The maximum is at $\Re\langle\psi_1,\psi_2\rangle=+\sqrt{E_1E_2}$:
\[
S_{\max}
= \frac{(E_1+E_2+2\sqrt{E_1E_2})\,\sqrt{E_1E_2}}{(E_1+E_2)^2}.
\]
Let $a=\sqrt{E_1}$, $b=\sqrt{E_2}$. Then
\[
S_{\max} = \frac{ab(a+b)^2}{(a^2+b^2)^2} \;\le\; 1,
\]
since $(a^2+b^2)^2 - ab(a+b)^2 = (a-b)^2(a^2+ab+b^2)\ge 0$.
Thus $0\le S\le 1$.
\end{proof}

\begin{proposition}[Symmetry and self-match]
$S(\psi_1,\psi_2)=S(\psi_2,\psi_1)$ and $S(\psi,\psi)=1$ whenever $\psi\neq 0$.
\end{proposition}

\begin{proof}
Symmetry is immediate from \eqref{eq:S-compact}.
For self-match, set $\psi_2=\psi_1$: then $E_1=E_2$ and $\Re\langle\psi_1,\psi_1\rangle=E_1$,
so
\(
S(\psi,\psi) = \frac{(2E_1+2E_1)E_1}{(2E_1)^2} = 1.
\)
\end{proof}

\begin{proposition}[Global phase invariance]
For any $\delta\in\mathbb{R}$, $S(e^{i\delta}\psi_1, e^{i\delta}\psi_2)=S(\psi_1,\psi_2)$.
\end{proposition}

\begin{proof}
$E_k$ are invariant under a global phase.
Moreover,
\(
\langle e^{i\delta}\psi_1, e^{i\delta}\psi_2\rangle
= \sum e^{i\delta}\psi_1 \overline{e^{i\delta}\psi_2}
= \sum e^{i\delta}\psi_1 e^{-i\delta}\overline{\psi_2}
= \langle\psi_1,\psi_2\rangle.
\)
Apply \eqref{eq:S-compact}.
\end{proof}

\begin{proposition}[Anti-phase minimum]
If $\psi_2=-\psi_1$ then $S(\psi_1,\psi_2)=0$.
\end{proposition}

\begin{proof}
Then $E_1=E_2$ and $\Re\langle\psi_1,-\psi_1\rangle=-E_1$,
so the numerator in \eqref{eq:S-compact} vanishes.
\end{proof}

\begin{theorem}[Reduction to cosine in the real, equal-norm case]
Assume $\psi_1,\psi_2\in\mathbb{R}^L$ and $E_1=E_2>0$.
Let $\theta$ be the angle between them.
Then $S(\psi_1,\psi_2)=\frac{1+\cos\theta}{2}$.
\end{theorem}

\begin{proof}
For real vectors, $\Re\langle\psi_1,\psi_2\rangle = \|\psi_1\|_2\|\psi_2\|_2\cos\theta$.
With $E_1=E_2$, one has $\sqrt{E_1E_2}=E_1$ and $E_1+E_2=2E_1$,
hence \eqref{eq:S-compact} gives
\(
S = \frac{(2E_1+2E_1\cos\theta)E_1}{(2E_1)^2} = \frac{1+\cos\theta}{2}.
\)
\end{proof}

\begin{proposition}[When does $S=1$?]
$S(\psi_1,\psi_2)=1$ iff $\psi_1=\psi_2\neq 0$.
\end{proposition}

\paragraph{Remark (Not a metric).}
$S$ is a bounded similarity, not a distance; $1-S$ does not in general satisfy the triangle inequality.

\bibliographystyle{unsrt}
\bibliography{references}
\end{document}